\documentclass{article}
\usepackage{ijcai17}

\usepackage{times}
\usepackage[font=small]{caption}

\usepackage[ruled,noend]{algorithm2e}

\usepackage{amsmath}
\usepackage{amsthm}
\usepackage{paralist}
\usepackage{multicol}
\usepackage{amsfonts}
\usepackage{breqn}
\usepackage{float}

\usepackage{tikz}
\usetikzlibrary{fit,arrows,calc,positioning,shapes.geometric}
\usepackage{tikz-qtree}
\usepackage{standalone}
\usepackage{pgfplots}
\pgfplotsset{compat=1.13}
\usepackage{environ}
\makeatletter
\newsavebox{\measure@tikzpicture}
\NewEnviron{scaletikzpicturetowidth}[1]{%
  \def\tikz@width{#1}%
  \begin{lrbox}{\measure@tikzpicture}%
  \BODY
  \end{lrbox}%
  \pgfmathparse{#1/\wd\measure@tikzpicture}%
  \BODY
}
\makeatother

\usepackage{mathtools}

\newtheorem{theorem}{Theorem}

\newtheorem{lemma}[theorem]{Lemma}

\newtheorem{definition}{Definition}
\newtheorem{fact}{Fact}

\def\e{\ensuremath{{\xi}}}

\newcommand{\bbI}{\ensuremath{\mathbb{I}}}

\newcommand{\cD}{\ensuremath{\mathcal{D}}}

\newcommand{\cX}{\ensuremath{\mathcal{X}}}
\newcommand{\cY}{\ensuremath{\mathcal{Y}}}

\newcommand{\egt}{EGT}
\newcommand{\efg}{EFG}

\def\epsilonsad{{\epsilon_{\hbox{\scriptsize\rm sad}}}}

\newcommand{\be}{\begin{eqnarray}}
\newcommand{\ee}[1]{\label{#1}\end{eqnarray}}

\newcommand{\ese}{\end{eqnarray*}}
\newcommand{\bse}{\begin{eqnarray*}}
\def\beq{\begin{equation}}
\def\eeq{\end{equation}}

\def\fnote#1{\footnote}

\def\*{{{\LARGE\bf $^*$}}}

\def\R{{\mathbb{R}}}

\def\cD{{\cal D}}

\def\cX{{\cal X}}
\def\cY{{\cal Y}}

\def\ri{{\mathop{\rm ri}\,}}

\def\log{\mathop{{\rm log}}}

\newcommand{\dgf}{DGF}

\title{Smoothing Method for Approximate Extensive-Form Perfect Equilibrium}
\author{Christian Kroer \and Gabriele Farina \and Tuomas Sandholm\\
Computer Science Department\\ Carnegie Mellon University\\
{\{ckroer,gfarina,sandholm\}@cs.cmu.edu}
}

\begin{document}

\maketitle

\begin{abstract}
Nash equilibrium is a popular solution concept for solving
imperfect-information games in practice. However, it has a major drawback: it
does not preclude suboptimal play in branches of the game tree that are not
reached in equilibrium. Equilibrium refinements can mend this issue, but have
experienced little practical adoption. This is largely due to a lack of
scalable algorithms.

Sparse iterative methods, in particular first-order methods, are known to be
among the most effective algorithms for computing Nash equilibria in large-scale
two-player zero-sum extensive-form games. In this paper, we provide, to our
knowledge, the first extension of these methods to equilibrium refinements. We
develop a smoothing approach for behavioral perturbations of the convex polytope
that encompasses the strategy spaces of players in an extensive-form game. This
enables one to compute an approximate variant of extensive-form perfect
equilibria. Experiments show that our smoothing approach leads to solutions with
dramatically stronger strategies at information sets that are reached with low
probability in approximate Nash equilibria, while retaining the overall
convergence rate associated with fast algorithms for Nash equilibrium. This has
benefits both in approximate equilibrium finding (such approximation is
necessary in practice in large games) where some probabilities are low while
possibly heading toward zero in the limit, and exact equilibrium computation
where the low probabilities are actually zero.
\end{abstract} 
\section{Introduction}
Nash equilibrium is the basic solution concept for noncooperative
games, including \emph{extensive-form games (EFGs)}, a broad class of games that model
sequential and simultaneous interaction, imperfect information, and outcome
uncertainty~\cite{Sandholm10:State,Bowling15:Heads-up,Brown15:Hierarchical,Moravcik17:DeepStack}.
Nash equilibrium was the solution concept used in the \emph{Libratus} agent, which
showed superhuman performance against a team of top Heads-Up No-Limit Texas
hold'em poker specialist professional players in the \emph{Brains vs. AI} event in January 2017~\cite{Brown17:Safe_arxiv}. It was
also used in the \emph{DeepStack} agent \cite{Moravcik17:DeepStack}, which beat a
group of professional players. It has also been dominant in the \emph{Annual
  Computer Poker Competition}~[ACPC], where the winning agents 
have all been based on Nash equilibrium approximation for many years.

In spite of this popularity, Nash equilibria suffer from a major deficiency:
they might not play reasonably in parts of the game tree that are reached with
zero probability in equilibrium. In particular, the only guarantee that Nash
equilibrium gives in these parts of the game tree is that it does not give up
more utility than the value of the game. Thus, if the opponent makes a big
mistake, Nash equilibrium might give back all the utility gained from the
opponent making that mistake, since it is only maintaining the value of the
game~(Miltersen and S\o rensen~\shortcite{Miltersen10:Computing} show nice examples of such
behavior).

The above shows that Nash equilibrium is not satisfactory in extensive-form
games, and is the motivation for equilibrium
refinements~\cite{Selten75:Reexamination}. When information is perfect, the
classical solution concept of \emph{subgame-perfect equilibrium (SPE)} can be
satisfactory, while it is not when information is imperfect. In this latter
case, refinements are usually based on the idea of perturbations representing
mistakes of the players. In a \emph{quasi-perfect equilibrium
  (QPE)}~\cite{VanDamme84:Relation}, a player maximizes her utility in each
decision node taking into account the future mistakes of the opponents only,
whereas, in an \emph{extensive-form perfect equilibrium (EFPE)}, players maximize
their utility in each decision node taking into account the future mistakes of
both themselves and their
opponents~\cite{Selten75:Reexamination,Hillas02:Foundations}.

Computation of Nash equilibrium
refinements in {EFG}s has received some attention in the literature. Von Stengel et. al.~\shortcite{Stengel02:Computing_b}
give a pivoting algorithm for computing normal-form-perfect equilibria in
{EFG}s. Miltersen and S\o rensen~\shortcite{Miltersen10:Computing} give an algorithm
for computing quasi-perfect equilibria. Miltersen and S\o rensen~\shortcite{Miltersen08:Fast} show how to compute a normal-form-proper
equilibrium. Farina and Gatti~\shortcite{Farina17:Extensive_Form} give an
algorithm for computing extensive-form perfect equilibria. All these results rely
on linear programming (LP) (in the zero-sum case) or linear complementary
programming (LCP).
In zero-sum games, several of these solution concepts can be computed in
polynomial time using an LP or a series of LPs. However, even for the easier
case of Nash equilibria, the LP approach is not scalable for large games
(beyond roughly $10^8$ nodes in the game tree~\cite{Gilpin07:Lossless}).
Each iteration of an LP-solving algorithm is expensive, and the LP might even be
too large to fit in memory.
In practice, iterative methods are preferred, even for games of modest size.
These methods have iteration costs that are usually linear, or better, in the
game size, but converge to a Nash equilibrium only in the limit. The most
prominent of these methods are \emph{counterfactual regret minimization (CFR)}~\cite{Zinkevich07:Regret} and its
variants~\cite{Lanctot09:Monte,Tammelin15:Solving,Brown15:Regret-Based,Brown17:Reduced}, and
general first-order methods (FOMs) such as the \emph{excessive gap technique}
(EGT)~\cite{Nesterov05:Excessive}  instantiated with an appropriate EFG
\emph{smoothing
  technique}~\cite{Hoda10:Smoothing,Kroer15:Faster,Kroer17:Theoretical}.
Farina \textit{et al.}~\shortcite{Farina17:Regret} show how to extend CFR to approximate EFPEs.

In this paper, we show how to extend FOMs to the computation of an approximate
variant of EFPE. Miltersen and S\o rensen~\shortcite{Miltersen10:Computing} and Farina and Gatti~\shortcite{Farina17:Extensive_Form} presented perturbed polytopes of EFGs that capture equilibrium refinements where each action has to be played with positive probability. We prove that recent results on smoothing techniques for EFGs based on dilating the entropy function can be modified to provide
smoothing for such perturbed games, where the perturbations are with respect to
\emph{behavioral strategies}. We then instantiate this method for the perturbed game of Farina and Gatti, which leads to our approximate EFPE.

We then experimentally validate our method. We show that it is effective at obtaining low maximum regret at each information set of the game---even ones that have low probability of being reached---while simultaneously achieving the same practical convergence rate that FOMs and the best CFR variants traditionally achieve for just Nash equilibrium. This has benefits both in approximate Nash equilibrium finding (such approximation is necessary in practice in large games) where some probabilities are low while possibly heading toward zero in the limit, and exact Nash equilibrium computation where the low probabilities are actually zero.

\section{Preliminaries}

We assume that the reader is familiar with the classical concept of extensive-form game. We invite the reader unfamiliar with the topic to refer to Shoham and Leyton-Brown~\shortcite{Shoham08:Multiagent} or any classic textbook on the subject for further information and context. Briefly, an extensive-form game $\Gamma$ is defined over a game tree. In each non-terminal node a single player moves and each edge corresponds to an action available to the player. Each leaf node is associated with a payoff vector, representing the utility for the two players when the game finishes in the leaf.

A Nash equilibrium is defined in Definition~\ref{def:nash}.

\begin{definition}
  An $\epsilon$-NE is a strategy profile $(\pi_1, \pi_2)$ for the players, such that no player can gain more than $\epsilon$ by unilaterally deviating from their strategy.
\end{definition}

\begin{definition}\label{def:nash}
  A Nash equilibrium (NE) is a 0-NE.
\end{definition}

However, Nash equilibria might not be satisfactory when
dealing with {\efg}s, independently of whether the game has perfect
or imperfect information, and whether it is general- or zero-sum. A Nash
equilibrium $\pi$ might prescribe irrational play in those information sets that
are visited with zero probability when playing according to $\pi$
(e.g.,~\cite{Miltersen08:Fast}). In the general-sum case, consider the left example of Figure~\ref{fig:nash inappropriate}: the strategy profile $(\pi_1,\pi_2)$ where player 1 always chooses action \textsl{x} and player 2 always chooses action \textsf{y} is a NE. However, this strategy profile is irrational: Player 2 is ``threatening'' to play a suboptimal action, and Player 1 is caving in to the threat. Yet, the threat is not credible: if Player 1 were to actually play action \textsf{y}, it would be irrational for Player 2 to honor the threat.

\begin{figure}[H]
 \centering \begin{tikzpicture}[scale=.8]
    \fill (0,0) circle (.9mm) node[right=3mm] {Player 1};
    \fill (-.9, -1) circle (.9mm) node[below=2mm] {$(1, 5)$};   
    \fill (1, -1) circle (.9mm);   
    \draw[ultra thick] (-.9, -1) --node[left]{\textsf{x}} (0,0 );
    \draw (0, 0) --node[right]{\textsf{y}} (1,-1) node[right=3mm] {Player 2};
    \fill (0.3, -2) circle (.9mm) node[below=2mm] {$(5,1)$};
    \fill (1.7, -2) circle (.9mm) node[below=2mm] {$(0, 0)$};
    \draw (0.3, -2) --node[left]{\textsf{x}} (1, -1);
    \draw[ultra thick] (1, -1) --node[right]{\textsf{y}} (1.7, -2);
  \end{tikzpicture}
  \hfill\begin{tikzpicture}[scale=.8]
   \fill (0,0) circle (.9mm) node[right=3mm] {Player 1};
   \fill (-.9, -1) circle (.9mm) node[below=2mm] {$(1, -1)$};   
   \fill (1, -1) circle (.9mm);   
   \draw[ultra thick] (-.9, -1) --node[left]{\textsf{x}} (0,0);
   \draw (0, 0) --node[right]{\textsf{y}} (1,-1) node[right=3mm] {Player 2};
   \fill (0.3, -2) circle (.9mm) node[below=2mm] {$(-5,5)$};
   \fill (1.7, -2) circle (.9mm) node[below=2mm] {$(0, 0)$};
   \draw (0.3, -2) --node[left]{\textsf{x}} (1, -1);
   \draw[ultra thick] (1, -1) --node[right]{\textsf{y}} (1.7, -2);
 \end{tikzpicture}
 \vspace{-3mm}
 \caption{General-sum (left) and zero-sum (right) games where Nash equilibrium prescribes irrational play. Numbers in parentheses denote the payoffs to Players 1 and 2.}
 \label{fig:nash inappropriate}
 \vspace{-3mm}
\end{figure}
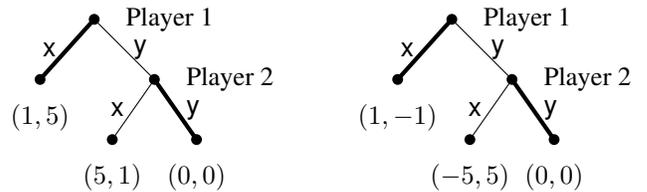

The right example in Figure~\ref{fig:nash inappropriate} shows that even in zero-sum games, a NE can fail to capture (sequential) rationality.
In this game, the same strategy profile as in the previous game is again a NE.
If Player 2 plays according to this profile, she gives up a potential payoff of 5 if Player 1 plays action
\textsf{y}. 

\subsection{Perturbations and Extensive-Form Perfection}
A way to mend the issue just described is to introduce the idea of ``trembling
hands'': each player cannot fully commit to a pure strategy, and ends up making
mistakes with a small (yet strictly positive) probability. This guarantees that
the whole game tree gets visited. More formally, let $l(h, a)$ be the
\emph{perturbation} of the game, a (positive) function defining the minimum
amount of probability mass with which the player playing at information set $h$
in the game will select action $a$ when playing in $h$. Let $\Gamma_l$ be the
game where players are subject to such perturbation: an extensive-form perfect
equilibrium of the game $\Gamma$ is any limit point of the sequence of Nash
equilibria of the game $\Gamma_l$, as $l$ vanishes~\cite{Selten75:Reexamination}.
In this paper, we deal with the simplest form of perturbation -- a uniform perturbation $l_\e$ for $\e > 0$, defined as $l_{\e}(h, a) = \e$ for all $a$ and $h$. We will denote the game $\Gamma_{l_\e}$ as $\Gamma_\e$.
\subsection{Bilinear Saddle-Point Problems and the Sequence Form}
It is well-known that the strategy spaces of an extensive-form game can be transformed into convex polytopes that allow a bilinear saddle-point formulation (BSPP) of the Nash equilibrium problem as follows~\cite{Romanovskii62:Reduction,Stengel96:Efficient,Koller96:Efficient}.
\begin{equation}
  \label{eq:sequence_form_objective} \min_{x \in \cX} \max_{y \in \cY} \langle
x, Ay \rangle = \max_{y \in \cY} \min_{x \in \cX} \langle x,Ay \rangle
\end{equation}
Our approach for computing equilibrium refinements will be based on constructing a perturbed variant of $\cX$ and $\cY$.


Several FOMs with attractive convergence properties
have been
introduced for BSPPs~\cite{Nesterov05:Smooth,Nesterov05:Excessive,Nemirovski04:Prox,Chambolle11:First}.
These methods rely on having some appropriate distance measure over
$\cX$ and $\cY$, called a \emph{distance-generating function} (DGF). Generally, FOMs
use the DGF to choose steps: given a gradient and a scalar stepsize, a FOM moves
in the negative gradient direction by finding the point that minimizes the sum of the gradient and of the DGF evaluated at the new point. In other words, the next step can be found by solving a regularized optimization problem, where long gradient steps are discouraged by the DGF. For EGT on EFGs, the DGF can be
interpreted as a smoothing function applied to the best-response problems faced
by the players.

\begin{definition}
  A distance-generating function for $\cX$ is a function $d(x):\cX \rightarrow
  \R$ which is convex and continuous on $\cX$, admits continuous selection of
  subgradients on the set $\cX^\circ=\left\{ x\in\cX: \partial d(x) \ne \emptyset
  \right\}$, and is strongly convex modulus $\varphi$ w.r.t. $\|\cdot\|$.
  Distance-generating functions for $\cY$ are defined analogously.
\end{definition}

Given a twice differentiable function $f$, we let $\nabla^2f(z)$ denote its Hessian at $z$. Our analysis is based on the following sufficient condition for strong convexity of a twice differentiable function:
\begin{fact} \label{fac:strong_convexity_hessian}
  A twice-differentiable function $f$ is strongly convex with modulus $\varphi$ with respect to a norm $\|\cdot\|$ on nonempty convex set $C\subset\R^n$ if 
$
  h^\top \nabla^2f(z) h \geq \varphi\|h\|,\ \forall h\in\R^n, z\in C^\circ.
$
\end{fact} 


Given DGFs $d_{\cX},d_{\cY}$ for $\cX,\cY$ with strong convexity moduli $\varphi_{\cX}$ and $\varphi_{\cY}$ respectively, we now describe the Excessive Gap Technique (EGT)~\cite{Nesterov05:Excessive} applied to \eqref{eq:sequence_form_objective}. EGT forms two
smoothed functions using the DGFs
\begin{align}
\vspace{-1mm}
  f_{\mu_y}(x) = \max_{y\in \cY} \langle x, Ay \rangle - \mu_\cY d_\cY, \label{eq:smoothed_y}\\
  \phi_{\mu_x}(y) = \min_{x\in \cX} \langle x, Ay \rangle + \mu_\cX d_\cX .\label{eq:smoothed_x}
\vspace{-1mm}
\end{align}
These functions are smoothed approximations to the optimization problem faced by
the $x$ and $y$ player, respectively. The scalars $\mu_1,\mu_2>0$ are smoothness
parameters denoting the amount of smoothing applied. Let $y_{\mu_2}(x)$ and
$x_{\mu_1}(y)$ refer to the $y$ and $x$ values attaining the optima in
\eqref{eq:smoothed_y} and \eqref{eq:smoothed_x}. These can be thought of as
\emph{smoothed best responses}. Nesterov~\shortcite{Nesterov05:Smooth} shows that the
gradients of the functions $f_{\mu_2}(x)$ and
$\phi_{\mu_1}(y)$ exist and are Lipschitz continuous. The gradient
operators and Lipschitz constants are given as follows

{\centering
  $\displaystyle\nabla f_{\mu_2}(x) = a_1 + Ay_{\mu_2}(x), \quad \nabla \phi_{\mu_1}(y) = a_2 + A^\top x_{\mu_1}(y),$
  $\displaystyle L_1\left(f_{\mu_2}\right) = \frac{\|A\|^2}{\varphi_\cY\mu_2}$ and $\displaystyle L_2\left(\phi_{\mu_1}\right) = \frac{\|A\|^2}{\varphi_\cX\mu_1}$.\\
}

Let the convex conjugate of $d:Q\rightarrow \R$ be denoted by $d^*(q)=\max_{q\in Q}g^Tq - d(q)$. Based on this setup, we formally state  {\egt}~\cite{Nesterov05:Excessive} as
Algorithm~1.

\begin{figure}
\centering
   \includegraphics[width=.9\linewidth]{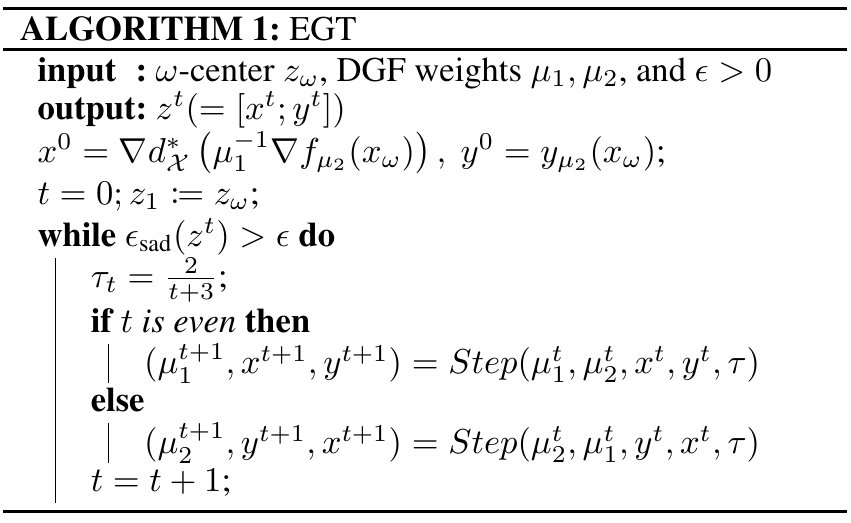}
   \includegraphics[width=.9\linewidth]{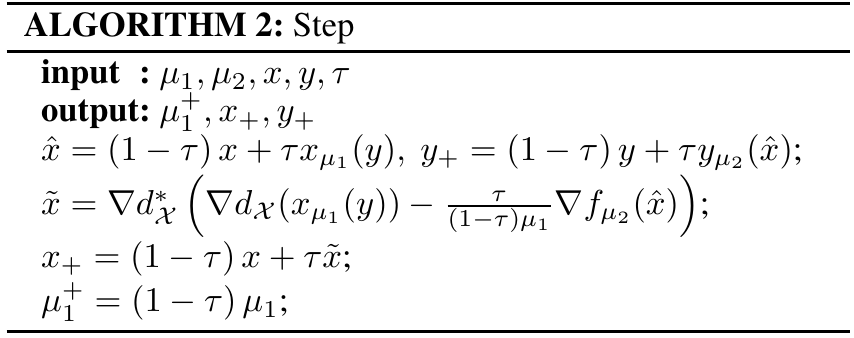}
\end{figure}

The \egt\ algorithm alternates between taking steps focused on $\cX$ and $\cY$. Algorithm~2 shows a single step focused on $\cX$. Steps focused on $y$ are analogous. Algorithm~1 shows how the alternating steps and stepsizes are computed, as well as how initial points are selected.

\noindent Suppose the initial values $\mu_1,\mu_2$ satisfy
$\mu_1=\frac{\varphi_\cX}{L_1(f_{\mu_2})}$. Then, at every iteration $t\geq 1$ of
 {\egt}, the corresponding solution $z^t=[x^t;y^t]$ satisfies
$x^t\in \cX$, $y^t\in \cY$, and
\[
\vspace{-1mm}
\max_{y\in\cY}(x^t)^TAy-\min_{x\in\cX}x^TAy^t= \epsilonsad(z^t) \leq  \frac{4\|A\|}{T+1}\sqrt{\frac{\Omega_\cX\Omega_\cY}{\varphi_\cX\varphi_\cY}}.
\vspace{-1mm}
\]
Consequently, {\egt} has a convergence rate of $O(\frac{1}{\epsilon})$~\cite{Nesterov05:Excessive}.

\subsection{Treeplexes}
\label{sec:treeplexes}
Hoda et al.~\shortcite{Hoda10:Smoothing} introduce the {\em treeplex}, a class of convex polytopes that captures the sequence-form of the strategy spaces in perfect-recall {\efg}s. 
\begin{definition}
  Treeplexes are defined recursively:
  \begin{enumerate}
  \item {\em Basic sets}: The standard simplex $\Delta_m$ 
  is a treeplex.
  \item {\em Cartesian product}:  If $Q_1,\ldots, Q_k$ are treeplexes, then $Q_1 \times \cdots \times Q_k$ is a treeplex.
  \item {\em Branching}: Given a treeplex $P\subseteq \left[ 0,1 \right]^p$, a collection of treeplexes $Q=\left\{ Q_1,\ldots,Q_k \right\}$ where $Q_j\subseteq \left[ 0,1 \right]^{n_j}$, and $l=\left\{l_1,\ldots,l_k \right\} \subseteq \left\{ 1,\ldots, p \right\}$, the set defined by 
\begin{align*}
\vspace{-1mm}
  P\framebox{l}Q \coloneqq \left\{ \left(x,y_1,\ldots,y_k\right) \in \R^{p+\sum_j n_j}  :~ x\in P, \right.\\[-1mm]
    \left. \, y_1\in x_{l_1} \cdot Q_1,\, \ldots, y_k\in x_{l_k} \cdot Q_k \vphantom{\R^\sum}\right\}
\vspace{-1mm}
\end{align*}
    is a treeplex. 
    We say $x_{l_j}$ is the branching variable for the treeplex $Q_j$.
  \end{enumerate}
\end{definition}

For a treeplex $Q$, we denote by $S_Q$ the index set of the set of simplexes contained in $Q$ (in an \efg\ $S_Q$ is the set of information sets belonging to the player). For each $j\in S_Q$, the treeplex rooted at the $j$-th simplex $\Delta^j$ is referred to as $Q_j$. Given vector $q\in Q$ and simplex $\Delta^j$, we let $\bbI_j$ denote the set of indices of $q$ that correspond to the variables in $\Delta^j$ and define $q^j$ to be the subvector of $q$ corresponding to the variables in $\bbI_j$.  For each simplex $\Delta^j$ and branch $i\in \bbI_j$, the set $\cD_j^i$ represents the set of indices of simplexes reached immediately after $\Delta^j$ by taking branch $i$ (in an \efg, $\cD_j^i$ is the set of potential next-step information sets for the player). Given a vector $q\in Q$,  simplex $\Delta^j$, and index $i\in \bbI_j$, each child simplex $\Delta^k$ for every $k\in \cD_j^i$ is scaled by $q_i$. For a given simplex $\Delta^j$, we let $p_j$ denote the index in $q$ of the parent branching variable $q_{p_j}$ scaling $\Delta^j$. We use the convention that $q_{p_j}=1$ if $Q$ is such that no branching operation precedes $\Delta^j$. For each $j\in S_Q$, $d_j$ is the maximum depth of the treeplex rooted at $\Delta^j$, that is, the maximum number of simplexes reachable through a series of branching operations at $\Delta^j$. Then $d_Q$ gives the depth of $Q$. We use $b_Q^j$ to identify the number of branching operations preceding the $j$-th simplex in $Q$. We say that a simplex $j$ such that $b_Q^j=0$ is a \emph{root simplex}.

Our analysis requires a measure of the size of a treeplex $Q$. Thus, we define
$M_Q\coloneqq \max_{q\in Q} \|q\|_1$. 

In the context of {\efg}s, suppose $Q$ encodes player 1's strategy space;  then $M_Q$ is the maximum number of information sets with nonzero probability of being reached when player 1 has to follow a pure strategy  while the other player may follow a mixed strategy. 
We also let 
\begin{equation}\label{eq:max_norm_cutoff}
\vspace{-1mm}
M_{Q,r}\coloneqq \max_{q\in Q} \sum_{j\in S_Q: b_Q^j \leq r} \|q^j\|_1. 
\vspace{-1mm}
\end{equation}
Intuitively, $M_{Q,r}$ gives the maximum value of the $\ell_1$ norm of any vector $q\in Q$ after removing the variables corresponding to simplexes that are not within $r$ branching operations of the root of $Q$.

We let $Q^\e$ refer to a $\xi$-perturbed variant of a treeplex $Q$, for the perturbed game $\Gamma_\xi$.
$Q^\e$ is the intersection of $Q$ with the set of constraints $q^j \geq \e
q_{p_j}$ for all $j\in S_Q$. By constructing perturbed polytopes $\cX^\e,\cY^\e$
and using these rather than $\cX,\cY$ in \eqref{eq:sequence_form_objective}, we
get an approximate variant of EFPEs.

\section{Distance-Generating Functions for the $\xi$-Perturbed Game}

Let $d_s$ be a DGF for the $n$-dimensional simplex $\Delta_n$. We construct a DGF for $Q$ by \emph{dilating} $d_s$ for each simplex in $S_Q$ and take their sum:
$
    d(q) = \sum_{j\in S_Q} \beta_j q_{p_j}d_s(\frac{q^j}{q_{p_j}})
$. This class of DGFs for treeplexes was introduced by Hoda et. al.~\shortcite{Hoda10:Smoothing} and has been further studied by Kroer et. al.~\shortcite{Kroer15:Faster,Kroer17:Theoretical}.
We show that $d_s$ and $d$ can be used to implement a smoothing function for
$Q^\e$ and reason about its properties. To construct a smoothing function for
$Q^\e$, we first construct a smoothing function for an \emph{$\e$-perturbed
  simplex} $\Delta_n^\e=\left\{ q^s : \|q^s\|_1=1, q^s \geq \e \right\}$, with $\e >
0$. We construct a smoothing function for $\Delta_n^\e$ by composing $d_s$ with
a simple affine mapping $\phi({\tilde q^s})=\frac{{\tilde q^s} - \e}{1-n\e}$, which sets up a
one-to-one mapping between $\Delta_n$ and $\Delta_n^\e$. The inverse of this
function is $\phi^{-1}(q^s) = (1-n\e)q^s + \e$. We let $d_s^\e=d_s(\phi({\tilde q^s}))$. We
will show that $d_s^\e$ retains all nice {\dgf} properties of $d_s$.

Since $d_s$ is continuously differentiable, we can apply the chain rule to get
\begin{align}
\label{eq:d_eps}  
\nabla d_s^\e(q^s) = (1-n\e)^{-1} \nabla d_s({\tilde q^s}).
\end{align}
For our new {\dgf} to be practical we need the conjugate and its gradient to be
easily computable. We show that this reduces to a simple transformation of the
conjugate of $d_s$:
\begin{lemma}
  \label{lem:simplex_prox}
  For a simplex {\dgf} $d_s$ and its $\e$-perturbed variant $d_s^\e$, the convex conjugate and its gradient for $d_s^\e$ can be computed as
  \begin{equation*}
  \vspace{-1mm}
    d_s^{\e,*}(g) = d_s^*((1-n\e)g) + \langle g,\e \rangle
  \end{equation*}
  \begin{equation*}
    \nabla d_s^{\e,*}(g) = (1-n\e)\nabla d_s^*((1-n\e)g) + \e
    \vspace{-1mm}
  \end{equation*}
\end{lemma}
\begin{proof}
Follows by the definition of conjugate and the chain rule for gradients.
\end{proof}
Thus computing our conjugate reduces to computing the conjugate for $d_s$
coupled with simple linear transformations. Hoda et.
al.~\shortcite{Hoda10:Smoothing} showed that the conjugate for a treeplex based
on a sum over dilated simplex DGFs is easy to compute. Combined with
Lemma~\ref{lem:simplex_prox}, their result shows that the conjugate of a
treeplex DGF consisting of a sum over dilated perturbed simplex DGFs is easy to
compute, as long as the same holds for the individual conjugates.


We now focus on the case where $d_s$ is the entropy DGF for a simplex, that is, $d_s(q^s)=\sum_{i}q^s_i\log(q^s_i)$. Formally, we get the following DGF for a perturbed treeplex:
\begin{equation*}
\vspace{-1mm}
  d_Q^\e(q) = \sum_{j\in S_Q} \beta_j q_{p_j} \sum_{i\in \bbI_j} \frac{q_i/q_{p_j}-\e}{1-n_j\e_j}\log\left( \frac{q_i/q_{p_j}-\e}{1-n_j\e_j} \right)
\end{equation*}

Kroer et. al.~\shortcite{Kroer17:Theoretical} showed strong convexity and convergence results for the class of dilated entropy functions for treeplexes. We now show how their result can be leveraged to prove strong convexity bounds for the perturbed entropy DGF.
\begin{theorem}
  \label{the:strong_convexity}
  The dilated perturbed entropy DGF on a treeplex with weights that satisfy the following recurrence
  
  \begin{center}\begin{minipage}{.9\linewidth}
  \vspace{-1mm}
    \noindent$\displaystyle \alpha_j = 1 + \max_{i\in \bbI_j}\sum_{k \in \cD^i_j} \frac{\alpha_k\beta_k}{\beta_k - \alpha_k},\hfill \forall j \in S_Q,$\\
    \noindent$\displaystyle \beta_j > \alpha_j,\hfill\forall i\in \bbI_j \mbox{ and } \forall j \in S_Q ~\text{s.t.}~ b_Q^j > 0,$\\
    \noindent$\displaystyle\beta_j = \alpha_j,\hfill\forall i\in \bbI_j \mbox{ and } \forall j \in S_Q ~\text{s.t.}~ b_Q^j = 0.$\\
  \vspace{-1mm}
  \end{minipage}\end{center}

\noindent is strongly convex modulus $1$ with respect to the $\ell_2$ norm and modulus $\frac{1}{M_Q}$ with respect to the $\ell_1$ norm.
\end{theorem}
\begin{proof}
  We will show that the quadratic over the Hessian of $d_Q^\e$ can be expressed
  as a constant times the quadratic over the unperturbed dilated entropy DGF for
  $Q$. This will allow us to invoke the strong convexity theorem of Kroer et.
  al.~\shortcite{Kroer17:Theoretical}.

  Consider $q\in \ri(Q^\e)$ and any $h\in\R^n$. For each $j\in S_Q$ and
  $i\in\bbI_j$, the second-order partial derivates of $d_Q^\e(\cdot)$ with
  respect to $q_i$ are:
  \begin{align}
  \vspace{-1mm}
    \nabla_{q_i^2}^2 d_s^\e(q)=& \frac{\beta_j}{(1-n_j\e_j)(q_i-\e q_{p_j})}\nonumber \\
    &+ \sum_{k \in \cD^i_j}\sum_{l\in \bbI_{k}}\frac{\beta_{k}q_{l}^2}{(1-n_k\e_k)(q_l-\e q_i)q_i^2}\label{eq:partial_derivative_diagonal}
         \vspace{-1mm}
  \end{align}
  Also, for each $j\in S_Q,i\in\bbI_j$, the
  second-order partial derivates with respect to $q_i,q_{p_j}$ are given by:
  \begin{equation}
    \vspace{-1mm}
    \nabla_{q_i,q_{p_j}}^2 d_s^\e(q)= \nabla_{q_{p_j},q_{i}}^2 d_s^\e(q)
    = -\frac{\beta_{j}q_i}{(1-n_j\e)(q_i-\e q_{p_j})q_{p_j}}.
    \label{eq:partial_derivative_off_diagonal}
         \vspace{-1mm}
  \end{equation}
  Then equations~\eqref{eq:partial_derivative_diagonal} and~\eqref{eq:partial_derivative_off_diagonal} together imply
  \begin{align}
    \vspace{-1.5mm}
     & h^\top\nabla^2\omega(q)h 
       =  \sum_{j\in S_Q}\sum_{i\in \bbI_j} \left[ h_i^2 \left( \frac{\beta_j}{(1-n_j\e_j)(q_i-\e q_{p_j})}  \right.\right.\nonumber\\[-2mm]
     &\quad \left. + \sum_{k \in \cD^i_j} \sum_{l\in \bbI_{k}}\frac{\beta_{k}q_{l}^2}{(1-n_k\e_k)(q_l-\e q_i)q_i^2} \right) \nonumber\\[-2mm]
    &\quad \left.-   h_ih_{p_j}\frac{2\beta_{j}q_i}{(1-n_j\e)(q_i-\e q_{p_j})q_{p_j}}   \right] . \label{eq:hessian_quadratic_unsimplified} 
      \vspace{-1.5mm}
  \end{align}
Given $j\in S_Q$ and $i\in\bbI_j$, we have $p_k=i$ for each $k\in \cD_j^i$ and for any $k\in \cD_j^i$, there exists some other $j'\in S_Q$ corresponding to $k$  in the outermost summation. Then we can rearrange the following terms:
  \begin{align*}
    \vspace{-1mm}
    &\sum_{j\in S_Q}\sum_{i\in \bbI_j} h_i^2 \sum_{k \in \cD^i_j} \sum_{l \in \bbI_k} \frac{\beta_{k}q_{l}^2}{(1-n_k\e_k)(q_l-\e q_i)q_i^2}\\[-2mm]
    &= \sum_{j\in S_Q}\sum_{i \in \bbI_j} \beta_{j}\frac{h_{p_j}^2q_{i}^2}{(1-n_j\e_j)(q_i-\e q_{p_j})q_{p_j}^2}.
      \vspace{-1mm}
    \end{align*}
    Using this equality in \eqref{eq:hessian_quadratic_unsimplified} leads to
    \begin{align}
      \vspace{-1mm}
      &\eqref{eq:hessian_quadratic_unsimplified} = 
       \sum_{j\in S_Q}\sum_{i\in \bbI_j} \left[   \frac{\beta_jh_i^2}{(1-n_j\e_j)(q_i-\e q_{p_j})}  \right.\nonumber\\[-1mm]
      & \left.
        +  \frac{\beta_{j}h_{p_j}^2q_{i}^2}{(1-n_j\e_j)(q_i-\e q_{p_j})q_{p_j}^2}
        -   \frac{2\beta_{j}h_ih_{p_j}q_i}{(1-n_j\e)(q_i-\e q_{p_j})q_{p_j}} 
        \right] \nonumber\\[-1mm]
      & = \sum_{j\in S_Q}\sum_{i\in \bbI_j}\frac{\beta_jq_i
        \left(   \frac{h_i^2}{q_i} 
        +  \frac{h_{p_j}^2q_{i}}{q_{p_j}^2}
        -   \frac{2h_ih_{p_j}}{q_{p_j}} 
        \right)
        }{(1-n_j\e_j)(q_i-\e q_{p_j})}
        \label{eq:hessian2}
          \vspace{-2mm}
    \end{align}

    Now we can view the three terms inside the brackets as a convex function of $h_i$.
    First-order optimality implies that this function is nonnegative. 
    Furthermore, since $q_i \ge \e q_{p_j}$ we have $\frac{q_i}{q_i-\e q_{p_j}}\ge 1$. Combined, this gives
    \begin{align}
         \vspace{-1mm}
      \eqref{eq:hessian2}
       &\ge \sum_{j\in S_Q}\sum_{i\in \bbI_j}\frac{\beta_j
}{(1-n_j\e_j)}
        \left(   \frac{h_i^2}{q_i} 
        +  \frac{h_{p_j}^2q_{i}}{q_{p_j}^2}
        -   \frac{2h_ih_{p_j}}{q_{p_j}} 
        \right) \nonumber \\[-1mm]
       &\ge \sum_{j\in S_Q}\beta_j
         \left[\sum_{i\in \bbI_j}
        \left(   \frac{h_i^2}{q_i} 
        -   \frac{2h_ih_{p_j}}{q_{p_j}} 
        \right)
         +  \frac{h_{p_j}^2}{q_{p_j}}
         \right]
      \label{eq:hessian3}
         \vspace{-1mm}
    \end{align}
    The last step follows because $\frac{q_i}{q_{p_j}}$ form simplex weights. By
    Lemma~1 in Kroer et. al.~\shortcite{Kroer17:Theoretical} this is exactly the
    expression for the quadratic of the Hessian of the unperturbed dilated
    entropy function on $Q$ with weights $\beta_j$. Since our weights satisfy
    the requirements in Theorems~1 and~2 of Kroer et. al., the unperturbed
    dilated entropy function with these weights is strongly convex on $Q$, and
    thus we get \eqref{eq:hessian3} $\geq c\|h\|^2$ where $c=1$ when $\|\cdot\|$
    is the $l_2$ norm (by Theorem 1 of Kroer et. al.) and $c=\frac{1}{M_Q}$ when
    $\|\cdot\|$ is the $l_1$ norm (by Theorem 2 of Kroer et. al.). By
    Fact~\ref{fac:strong_convexity_hessian} this proves our theorem.
\end{proof}

Using Theorem~\ref{the:strong_convexity} we can use the perturbed dilated entropy function to instantiate EGT. Since the value of the perturbed entropy on $\Delta_n^\e$ can be lower-bounded by $\log(n)$ exactly the same way as with the unperturbed entropy, we can apply Theorem~3 of Kroer et. al.~\shortcite{Kroer17:Theoretical}, to bound EGT convergence rate as follows:

\begin{theorem}\label{the:entropy_diameter}
  For a perturbed treeplex $Q^\e$, the dilated perturbed entropy function with simplex weights $\beta_j=M_Q(2+\sum_{r=1}^{d_j}2^{r}(M_{Q_j,r}-1))$ for each $j\in S_Q$ results in  
$
  \frac{\Omega}{\varphi} \leq M_Q^2 2^{d_Q+2}\log m
$ where $m$ is the dimension of the largest simplex $\Delta^j$ for $j\in S_Q$ in the treeplex structure. 
\end{theorem}

Theorem~\ref{the:entropy_diameter} immediately leads to the following convergence rate result for EGT equipped with dilated perturbed entropy {\dgf}s to solve perturbed {\efg}s. 
\begin{theorem}\label{thm:D-EntropyRate}
  The {\egt} algorithm equipped with the dilated perturbed entropy {\dgf} with
  weights $\beta_{j}=2+ \sum_{r=1}^{d_j}2^{r}(M_{\cX_{j},r}-1)$ for all $j \in
  S_{\cX}$ and the corresponding setup for $\cY$ will return a
  $\epsilon$-accurate solution to the perturbed variant of
  \eqref{eq:sequence_form_objective} in at most the following number of
  iterations:
\[
\vspace{-1mm}
  \left(\max_{i,j}|A_{i,j}|\, \sqrt{M_\cX^22^{d_\cX+2}M_\cY^22^{d_\cY+2}}\, \log m\right) / \epsilon,
\vspace{-1mm}
\]
where the matrix norm is given by: \[
\|A\|=\max_{y\in \cY}\left\{ \| Ay \|_1^* :~ \|y\|_1=1\right\} = \max_{i,j}
|A_{i,j}|.\vspace{-1mm}\]
\end{theorem}
\noindent To our knowledge, this is the first result for FOMs that compute an approximate Nash equilibrium refinement.
\section{Experiments}
We conducted experiments to investigate the practical performance of our
smoothing approach when used to instantiate the EGT algorithm. We compare EGT
with our smoothing approach to EGT on an unperturbed polytope using the
smoothing technique by Kroer et. al.~\shortcite{Kroer17:Theoretical} and
CFR+~\cite{Tammelin15:Solving}. 
We conducted the experiments on Leduc hold'em poker~\cite{Southey05:Bayes}, a
widely-used benchmark in the imperfect-information game-solving community,
except we tested on a larger variant of the game in order to better test
scalability. In our enlarged version, \emph{Leduc 5}, the deck consists of $5$
pairs of cards $1\ldots 5$, for a total deck size of $10$. Each player initially
pays one chip to the pot, and is dealt a single private card. After a round of
betting, a community card is dealt face up. After a subsequent round of betting,
if neither player has folded, both players reveal their private cards. If either
player pairs their card with the community card they win the pot. Otherwise, the
player with the highest private card wins. In the event that both players have
the same private card, they draw and split the pot. Kroer et.
al.~\shortcite{Kroer17:Theoretical} point out that the theoretically sound scale
at which the overall weight on the {\dgf} should be set is too conservative. We
tune an overall weight on each {\dgf} by choosing the weight that performs best
with $EGT$ and $\e=0$ among $1,0.1,0.05,0.01,0.005$ on the first 20 iterations.
We test our approach on $\e$-perturbed polytopes of the strategy spaces
for $\e \in \left\{ 0.1,0.05, 0.01, 0.005, 0.001 \right\}$.

The first experiment measures convergence to Nash equilibrium
(Figure~\ref{fig:experiments_leduc5_eps}). The x-axis shows the number of tree
traversals performed per algorithm\footnote{Game tree traversals are equally
  expensive for all the algorithms studied. Treeplex traversal for each player is slower
  in EGT than CFR due to requiring exponentiation $\exp(\,\cdot\,)$, but
  the algorithms spend significantly less time on treeplex traversals than tree
  traversals, so this difference between the algorithms is insignificant.}. The y-axis shows the sum of player regrets in the full
(unperturbed) game. We find that the $\e$ perturbations have almost no effect on
overall convergence rate until convergence within the perturbed polytope, at
which point the regret in the unperturbed game stops decreasing, as expected.
This shows that our approach can be utilized in practice: there is no
substantial loss of convergence rate. Later in the run once the perturbed
algorithms have bottomed out, there is a tradeoff between exploitability in the
full game and refinement (i.e., better performance in low-probability
information sets).
\begin{figure}[]

%
     	\centering\includegraphics[height=4.5cm]{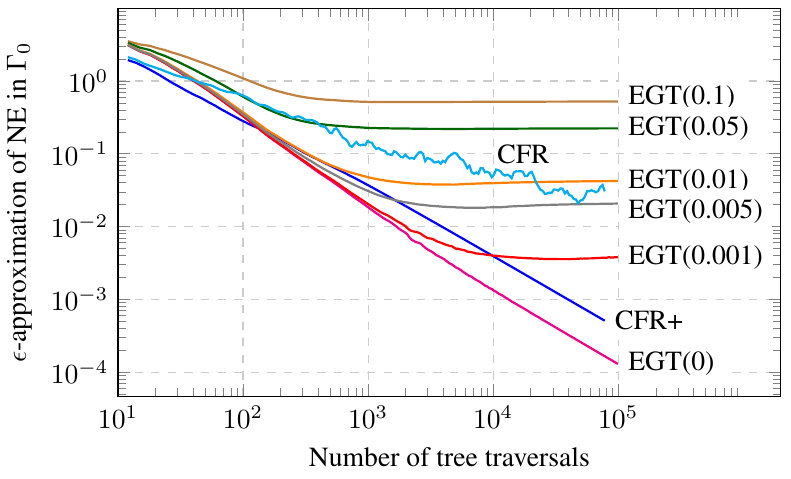}
    \caption{Regret as a function of the number of iterations for {\egt} with
      various $\e$ perturbations (denoted in parentheses) and CFR+. Both
      axes are on a log scale.}
  \label{fig:experiments_leduc5_eps}
\end{figure}

The second experiment shows a measure of refinement convergence (Figure~\ref{fig:experiments_leduc5_max_regret}). The x-axis shows the
number of tree traversals performed. The y-axis shows the
maximum regret at any individual information set. Information set regret
is calculated assuming that the information set is reached with probability one
and applying Bayes' rule to get a distribution over nodes at the information set;
the regret is the increase in expected utility from best-responding
throughout all information sets in the subtrees rooted at the information set.
Both CFR+ and unperturbed EGT perform badly with
respect to this measure of refinement. Both have maximum regret two orders of magnitude worse than the perturbed approach. The maximum regret one can possibly cause in an information set in Leduc 5 is 22, 
so CFR+ and unperturbed EGT also do poorly in that sense.
In contrast to this, we find that our $\e$-perturbed
solution concepts converge to a strategy with low regret at every information
set. The choice of $\e$ is important: for $\e=0.001$, the smallest perturbation,
we see that it takes a long time to converge at low-probability information
sets, whereas we converge reasonably quickly for $\e=0.01$ or $\e=0.005$; for
$\e=0.1$ and $\e=0.05$ the perturbations are too large, and we end up converging
with relatively high regret (due to being forced to play every action with
probability $\e$). Thus, within this set of experiments, $\e\in
\left[0.005,0.01\right]$ seems to be the ideal amount of perturbation. 

\begin{figure}[]
  \centering


	\includegraphics[height=4.5cm]{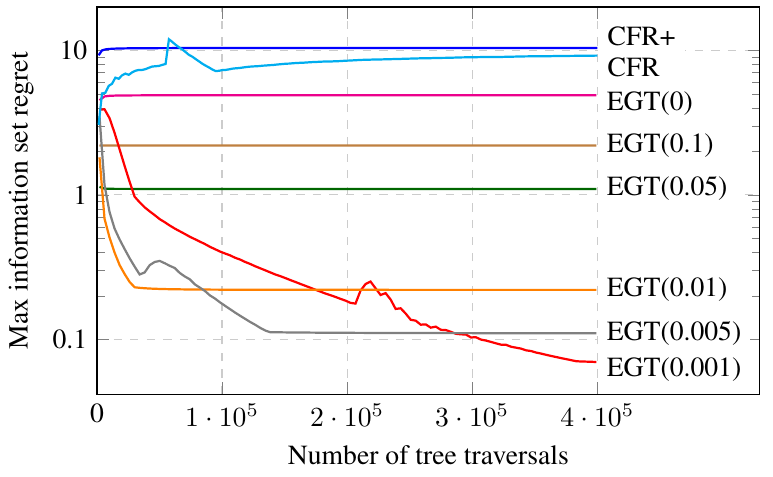}
    \caption{Maximum regret at any individual information set, as a function of
      the number of iterations.}
  \label{fig:experiments_leduc5_max_regret}
    \vspace{-1mm}
\end{figure}

\section{Conclusion and future research}

We studied the extension of FOMs to the computation of Nash-equilibrium
refinements. We developed a smoothing scheme based on perturbations of smoothing
schemes for standard EFG solving, and proved that the convergence rate is
comparable to that of solving the original game for Nash equilibrium. We
performed numerical simulations where we showed that our approach has an overall
convergence rate that is comparable to that of state-of-the-art Nash equilibrium
methods. At the same time, we showed that our approach leads to solutions that
have substantially better performance in subsets of the game tree that are
reached with low probability. This has benefits both in approximate Nash
equilibrium finding (such approximation is necessary in practice in large games)
where some probabilities are low while possibly heading toward zero in the
limit, and exact Nash equilibrium computation where the low probabilities are
actually zero.

Our work suggests several research directions. It
would be interesting to find a way to systematically decrease the $\e$-perturbations over time, so that we eventually converge to an exact Nash
equilibrium in the full game. This requires at least two extensions. First, FOMs usually assume static domains, whereas this would involve a
slowly expanding domain. Second, the $\e$ would need to be decreased at a rate
that is simultaneously fast enough that it converges at a reasonable rate, and
slow enough that we actually converge to a refinement.

We showed how to compute approximate EFPE refinements using methods that scale to large games. It would be interesting to find a way to instantiate scalable methods such as FOMs or CFR+ for other equilibrium refinement concepts as well. The perturbed polytope due to Miltersen and S\o rensen could be used to construct a notion of approximate QPE that would lead to an optimization setup similar to ours. However, this will require constructing a DGF for the perturbed-QPE polytope, which has $\e$-perturbations on the realization plans. Our approach relied on $\e$-perturbations to the behavioral strategies, and so it is likely that a different DGF class is needed to handle approximate QPE. 

\section*{Acknowledgments}
This work was supported by NSF grants
IIS-1617590, IIS-1320620, IIS-1546752 and ARO award
W911NF-17-1-0082.
The first author is supported by a Facebook Fellowship.

\clearpage
\small
\bibliographystyle{named}
\bibliography{dairefs}
[ACPC] \texttt{http://www.computerpokercompetition.org}

\end{document}